\documentclass[a4paper,11pt,oneside]{article}

\usepackage{amssymb, amsfonts, amsmath, amsthm} 

\usepackage{graphicx}

\usepackage[english]{babel}
\usepackage{inputenc}

\usepackage{bbm}
\usepackage[geometry]{ifsym} 
\usepackage{setspace} 
\usepackage[compact]{titlesec} 

\usepackage[affil-it]{authblk}

\newtheorem{defi}{Definition} 
\newtheorem{lemma}{Lemma}

\newcommand{\cA}{\mathcal{A}}
\newcommand{\cF}{\mathcal{F}}

\newcommand{\bR}{\boldsymbol R}
\newcommand{\bN}{\boldsymbol N}
\newcommand{\bX}{\bold{X}}

\newcommand{\bZ}{{\boldsymbol Z}}

\newcommand{\E}{E}
\newcommand{\ip}[2]{\langle #1,#2 \rangle} 
\newcommand{\ii}{\mathbbm{i}}

\author{{\L}ukasz Kidzi\'nski}
\affil{CHILI Laboratory, \'Ecole polytechnique f\'ed\'erale de Lausanne,\\ RLC D1 740, CH-1015, Lausanne, Switzerland}
\title{Functional time series}

\begin{document}
\maketitle
\begin{abstract}
The continuous advances in data collection and storage techniques allow us to observe and record real-life processes in great detail. Examples include financial transaction data, fMRI images, satellite photos, earths pollution distribution in time etc. Due to the high dimensionality of such data, classical statistical tools become inadequate and inefficient. The need for new methods emerges and one of the most prominent techniques in this context is functional data analysis (FDA).

The main objective of this article is to present techniques of the analysis of temporal dependence in FDA. Such dependence occurs, for example, if the data consist of a continuous time process which has been cut into segments, days for instance. We are then in the context of so-called functional time series.
\end{abstract}
\section{Introduction}

In this article we introduce foundations of functional time series and the frequency--domain analysis in this context. We address the article to larger audience, assuming only elementary knowledge of probability theory and algebra. Although, we try to keep the text accurate, in some fragments we sacrifice detailed investigation for intuitive argumentation, referring advanced readers to appropriate sources.

The manuscript is divided into two parts. In the first part, we introduce concepts from statistics and functional data analysis (FDA). We built upon basic ideas about continuous functions and probability. In the second part, we present some state-of-the-art results in linear models for functional objects.

\subsection{Motivation for statistics on functional data}

The main concern of statistics is to obtain essential information from a sample of observations $X_1,X_2,...,X_N$ from some space of objects. We are given a finite sample of size $N \in \bN$, where $\{X_i\}_{i\in \bZ}$ can be scalars (like heights of a sample of students in a school), vectors (like points on a target after throwing several darts) or more complex objects, like genotypes, fMRI scans, images or frames of a video.

Functional data analysis deals with observations which can be naturally expressed as functions. Figures \ref{fig:example1}, \ref{fig:example2} and \ref{fig:example3} present several cases, from various areas of science, which fit into the framework of functional data analysis.

\begin{figure}
\includegraphics[height=5cm]{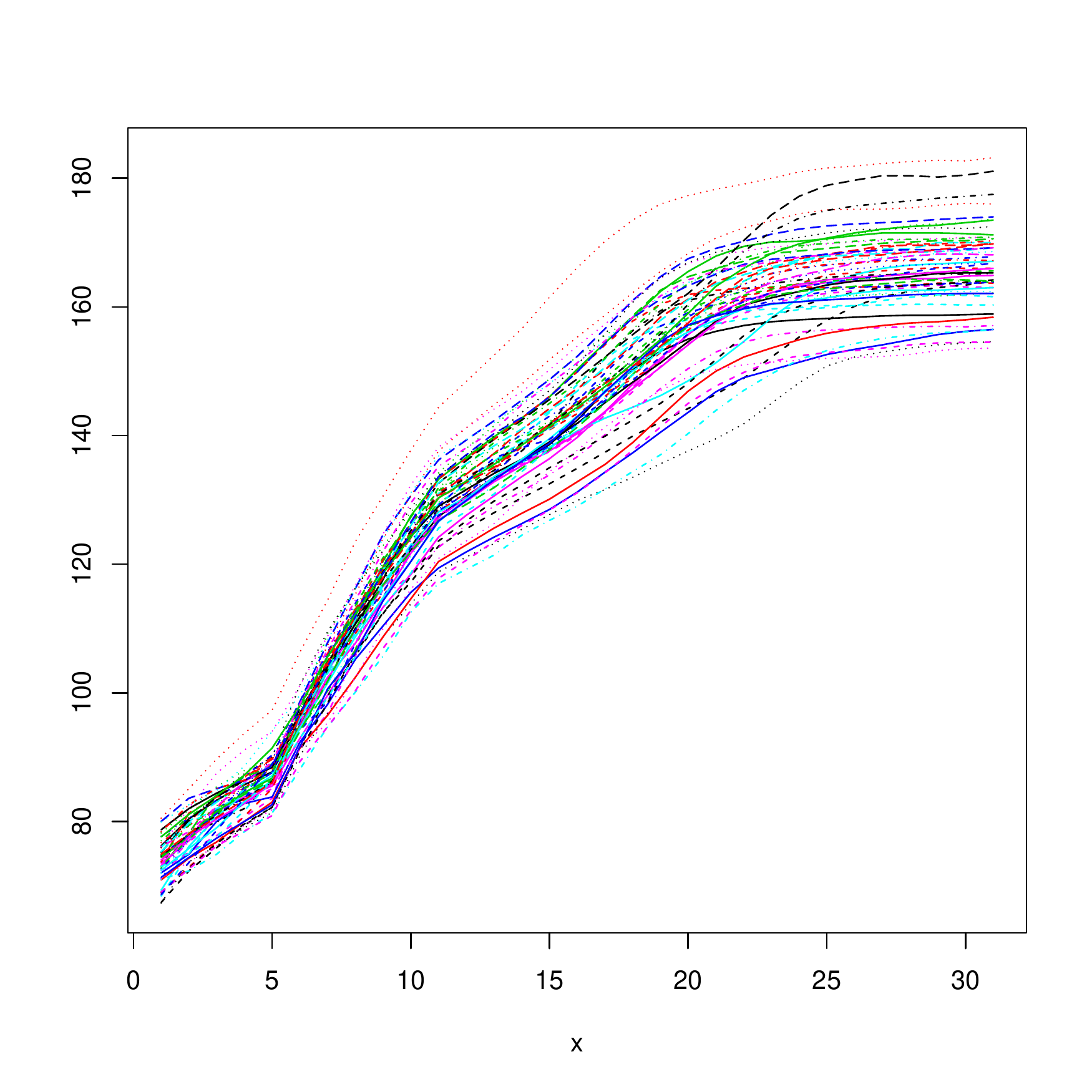}
\includegraphics[height=5cm]{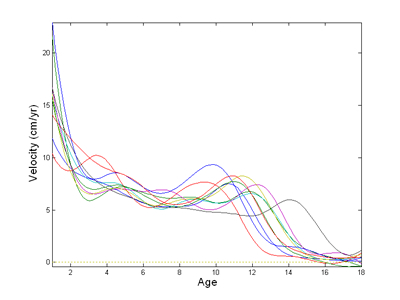}\centering
\caption{Berkeley Growth Data: Heights of 20 girls taken from ages 0 through 18 (left).
Growth process easier to visualize in terms of acceleration (right). Tuddenham and Snyder \cite{tuddenham1953physical} and Ramsey and Silverman \cite{ramsay:silverman:2005}}\label{fig:example1}
\includegraphics[height=6cm,width=10cm]{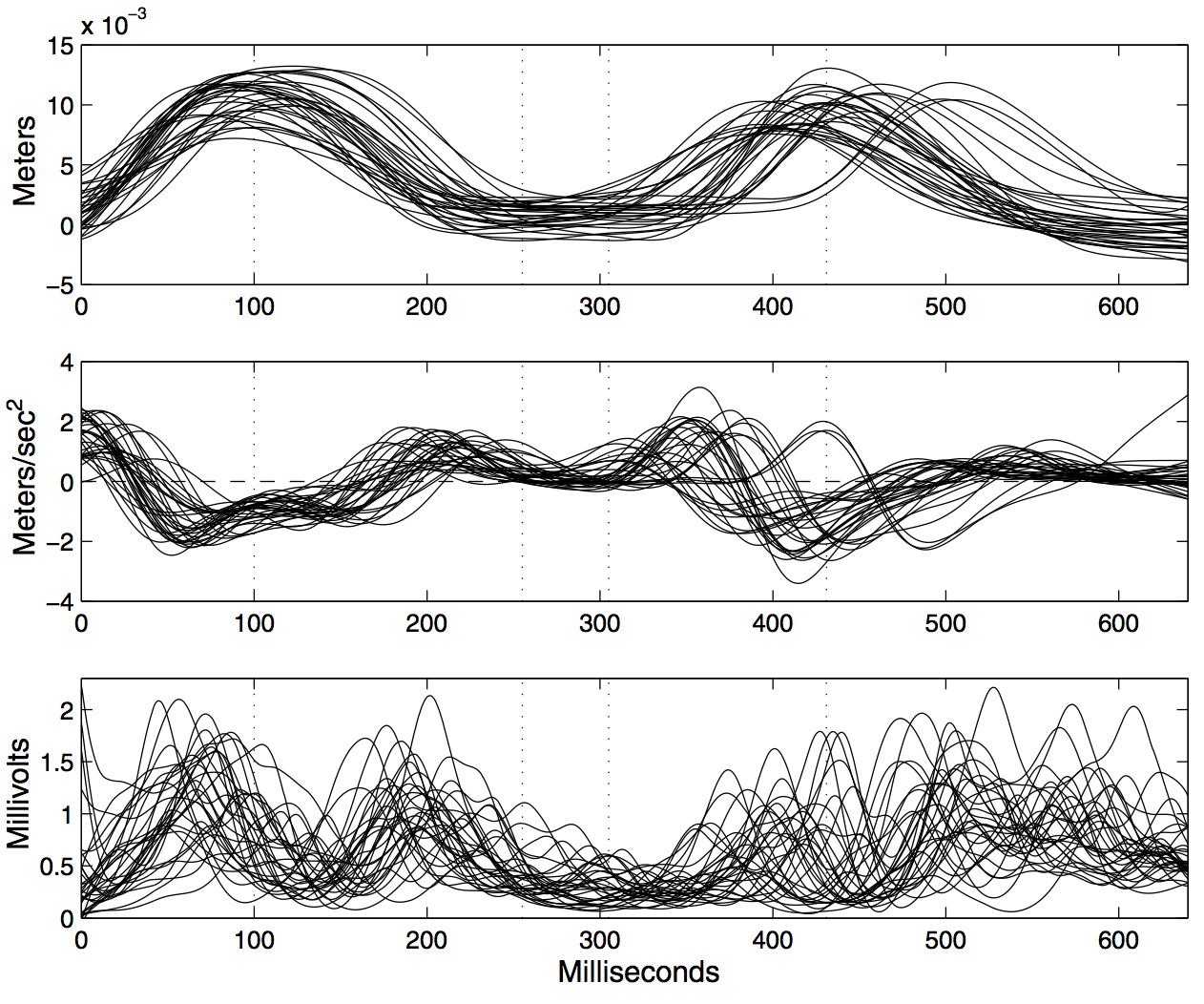}\centering
\caption{Lower lip movement (top), acceleration (middle) and EMG of a facial muscle (bottom)  of a speaker pronouncing the syllable ``bob'' for 32 replications. Malfait, Ramsay, and Froda \cite{malfait:ramsay:2003}}\label{fig:example2}
\includegraphics[height=6cm,width=10cm]{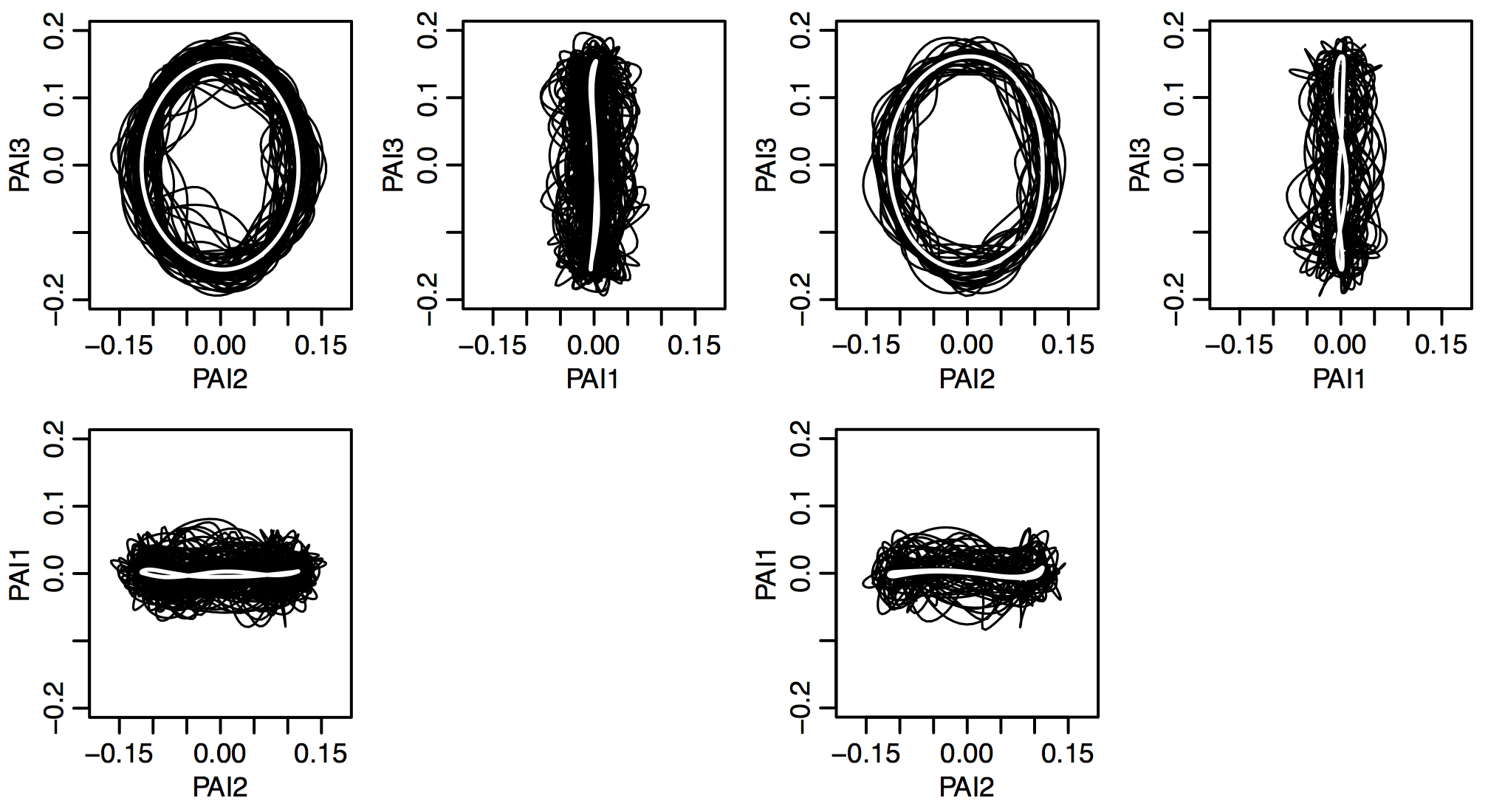}\centering
\caption{Projections of DNA minicircles on the planes given by the principal axes of inertia (three panels on the left side: TATA curves, right: CAP curves). Mean curves are plotted in white. Panaretos, Kraus and Maddocks \cite{panaretos2010second}}\label{fig:example3}
\end{figure}

When we deal with a physical process it is often natural to assume that it behaves in a continues manner and that the observations do not oscillate significantly between the measurements. Although, in the Digital Age, we rarely record analog processes continuously, we often have enough datapoints, so that interpolation doesn't cause a significant measurement error. Models incorporating this additional structure can lead to more precise and meaningful foundings. In this context, FDA can be seen as a tool which embeds the continuity feature into the model.

Except for good approximation of a continuous process, FDA can also prove to be useful in a noisy, discontinuous case. Then, FDA can serve as a tool for denoising and smoothing the data and is beneficial whenever the underlying process is the main concern, like, for example, in finance.

From a pragmatic perspective, functional data can be seen simply as infinitely dimensional vectors, with extended notion of variance and mean, and thus we may be tempted to employ classical multivariate techniques. However, there are many practical and theoretical problems that need to be addressed and this approach is not advised, as we will see later.

The FDA approach is also useful in a parsimonious representation of  the data by taking advantage of their smoothness. Instead of looking at a function as a dense vector of values, we can often represent it in an linear combination of a handful of (well--chosen) basis functions.

Practical applications of functional data analysis are spread across many areas of science and engineering. Panaretos et al.~\cite{panaretos2010second} use $[0,1] \rightarrow \bR^3$ closed curves to analyze the behavior of DNA microcircles, providing the testing methodology for the comparison of two classes of curves. Aston and Kirch~\cite{aston2013} analyze the stationarity and change point detection for functional time series, with applications to fMRI data. Hadjipantelis et al.~\cite{hadjipantelis2012} analyze Mandarin language using functional principal components.  Functional time series also naturally emerge in financial applications -- Kokoszka and Reimherr~\cite{kokoszka:reimherr:2013pred} analyze predictability of the shape of intraday price curves.  These works are only a fraction of the ongoing research and for a more accurate survey on applications and theory we refer to books \cite{ramsay:silverman:2005}, \cite{ferraty:vieu:2006}, \cite{HKbook} and \cite{bosq:2000}.

\subsection{Hilbert spaces}\label{sec:hilbert-spaces}

For most of the results presented in this work we only require a {\it separable Hilbert space}, a linear metric space with the norm function induced by the inner product and with a countable basis\footnote{In this work we introduce Hilbert spaces in an elementary and accessible way, avoiding technical details. For a formal definitions and investigation of key properties refer to \cite{rudin}.}. It makes the setup very general, but for simplicity, and in order to give an intuitive example to each of the results, in most cases we will assume a concrete space of square-integrable functions on a bounded interval $[0,1]$, to which we will refer to as $L_2$. A function $f : \bR \rightarrow \bR$ belongs to $L_2$ if and only if
\[
 \int_0^1 f^2(x) dx < \infty.
\]
In this section we present elementary properties of Hilbert spaces.

With a space $L_2$, we associate an {\it inner product}, a bilinear operator $L_2 \times L_2 \rightarrow \bR$. For two functions $f, g \in L_2$, we define the {\it inner product} as
\[
\ip{f}{g} = \int_0^1 f(x)g(x) dx.
\]
We define the norm of the element $f \in L_2$ as $\|f\| = \sqrt{|\ip{f}{f}|}$. Since $L_2$ is a space of square-integrable functions, each element $f \in L_2$ has a finite norm. We say that $f$ and $g$ are {\it orthogonal} if $\ip{f}{g} = 0$. If, additionally, $\|f\| = \|g\| = 1$, we say that $f$ and $g$ are {\it orthonormal}. Both definitions are related to vector spaces: a norm corresponds to the "distance" from $0$ function, whereas orthonormal elements behave as perpendicular vectors.

Hilbert space is called {\it separable} if we can find a series of pairwise orthonormal elements $e_1,e_2,e_3,...$ such that each element in $e \in L_2$ can be expressed as a weighted sum of elements $e_1,e_2,e_3,...$. This series $\{e_i\}_{1 \leq i}$ is called {\it a basis}, and each $e_i$ is {\it a basis function}.

One can show, that given the basis $\{e_i\}_{1 \leq i}$, the representation of $f \in L_2$ is uniquely given by
\begin{align}\label{basis:representation}
f = \sum_{i=1}^\infty \ip{e_i}{f} e_i, 
\end{align}
where scalars $\ip{e_i}{f}$ are called {\it coefficients in the basis $\{e_i\}_{1 \leq i}$}.

We may find infinitely many basis of $L_2$. As an example, often used in practice, consider the Fourier series, defined as
\[
e_i(x) = \begin{cases}
\sin(k \pi x), &\text{if } $i = 2k + 1$\\ 
\cos(k \pi x), &\text{if } $i = 2k$
\end{cases}
\]
where $k \in \bZ$ and $i \geq 1$. Several first elements are presented in Figure \ref{fig:fourier-basis}. A proof that elements are orthonormal is a simple exercise. The proof that each element in $f$ can be uniquely expressed as a linear combination \eqref{basis:representation} is more complicated and an interested reader is referred to \cite{rudin}.

\begin{figure}
\centering
\includegraphics[height=7cm]{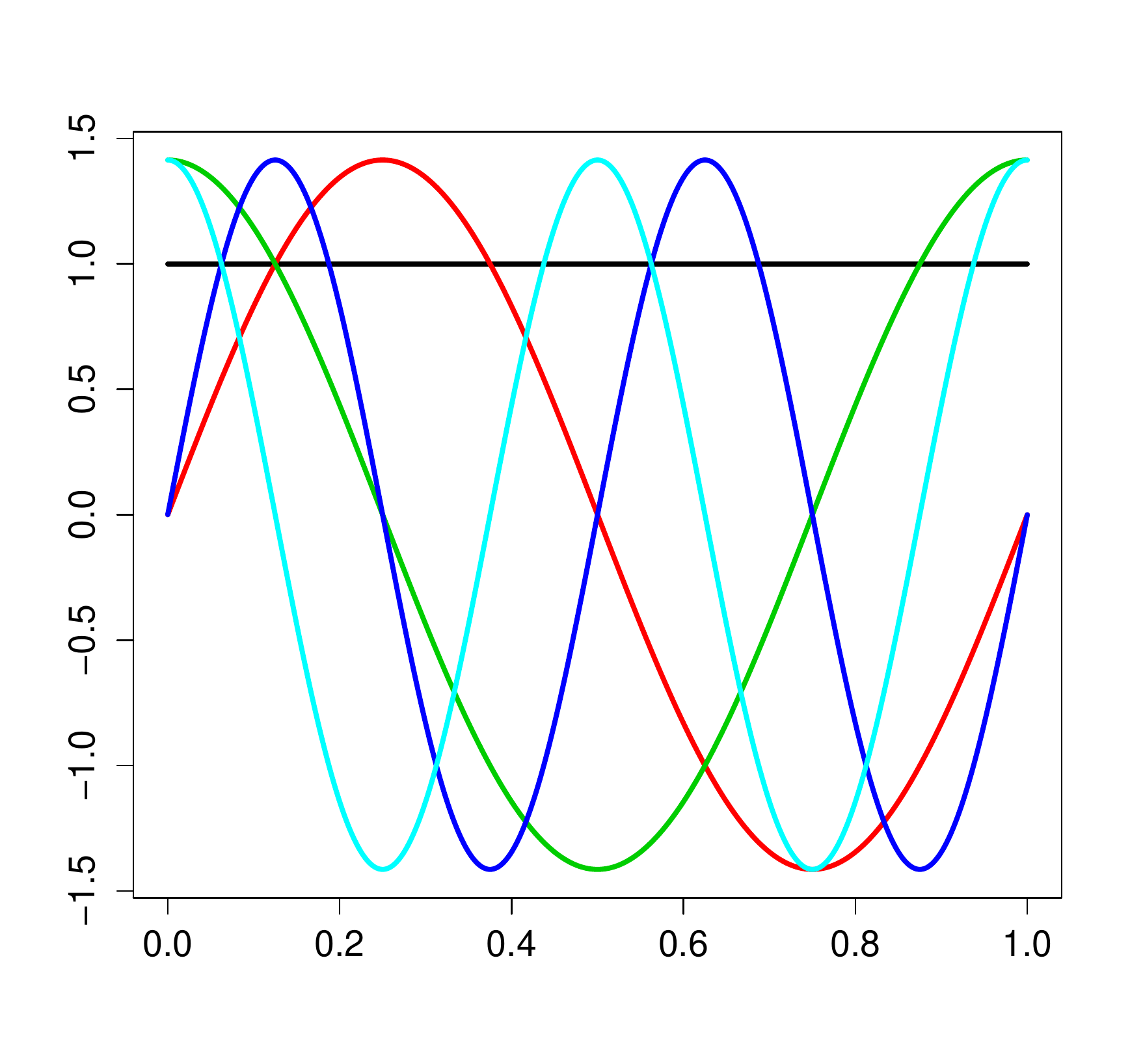}\ \ \ \ \ 
\includegraphics[height=7cm]{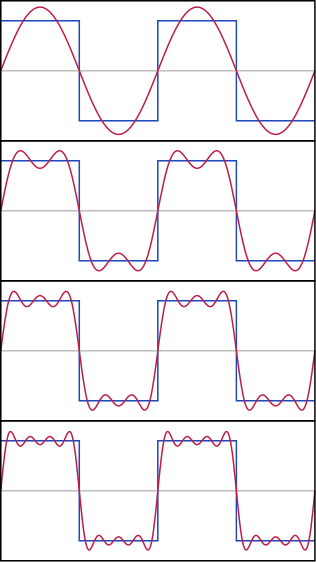}
\caption{First 5 Fourier basis functions in the order black, red, green, blue and light blue (left). Step function approximated by a weighted sum of $1,2,3$ and $4$ basis functions (right).}\label{fig:fourier-basis}
\end{figure}

We will use, one of the fundamental equations in separable Hilbert spaces, a Parseval's identity.
\begin{lemma}[Parseval's identity]
Let $f \in L_2$ and let $\{e_i\}_{1 \leq i}$ be an orthonormal basis in $L_2$. Then
\[
\|f\|^2 = \sum_{i = 1}^\infty  |\ip{e_i}{f}|^2.
\]
\end{lemma}
\begin{proof}
It follows from the definition of the norm
\begin{align*}
\|f\|^2 &= \langle f, f \rangle = \left\langle \sum_{i = 1}^\infty \ip{e_i}{f}e_i, \sum_{i = 1}^\infty \ip{e_i}{f}e_i \right\rangle\\
&=  \sum_{i = 1}^\infty \ip{e_i}{f} \left\langle e_i, \sum_{j = 1}^\infty \ip{e_j}{f}e_j \right\rangle\\
&=  \sum_{i = 1}^\infty \ip{e_i}{f} \sum_{j = 1}^\infty \ip{e_j}{f}\langle e_i, e_j \rangle\\
&=  \sum_{i = 1}^\infty  |\ip{e_i}{f}|^2,
\end{align*}
because $\ip{e_j}{f}\langle e_i, e_j \rangle = \ip{e_i}{f}$ if $i \neq j$ and $0$ otherwise, due to orthonormality of the series $\{e_i\}_{1 \leq i}$.
\end{proof}

We will also use the notion of {\it linear operators} and {\it Hilbert--Schmidt operators}. A {\it linear operator} $F~:~L_2~\rightarrow~L_2$ is a function such that for any pair of scalars $a,b$ and elements $v,w \in L_2$, $F(av + bw) = aF(v) + bF(w)$. A {\it Hilbert Schmidt--operator} $F$ is a linear operator, such that
\[
   \sum_{i=1}^\infty \|F(e_i)\|^2 < \infty,
\]
where $\{e_i\}_{i \geq 1}$ is an orthonormal basis of $L_2$. An operator is {\it symmetric} if $F(v) = F(-v)$ for each $v \in L_2$.

\subsection{Representation and fit}

In practice we are often given just a sample of observations from a curve and we need to interpolate, i.e. draw a curve between these points which is most likely to be close to the underlying process.

In this work, we follow the ideas popularized by Ramsey and Silverman \cite{ramsay:silverman:2005}, based on a basis function expansion \eqref{basis:representation}. Note that, given the representation \ref{basis:representation}, by the Parseval's identity, for any $\varepsilon > 0$, there exists $d$, such that
\[
\sum_{i = d}^\infty |\ip{e_i}{f}|^2 < \varepsilon.
\]
We can therefore approximate the function with arbitrary precision $\varepsilon > 0$ using only the first $d$ basis elements. 

In practice, inner products, which are typically obtained by integration, will be themselves approximated by corresponding sums. Then, a discretized sample curve $(x(t_j)\colon 1\leq j\leq n)$ can be transformed into a (finite dimensional) curve $y(t)$ through
$$
y(t):=\sum_{i=1}^d\left(\sum_{j=1}^nf(t_j)e_i(t_j)(t_j-t_{j-1})\right) e_i(t),
$$
for some grid $0 \leq t_1 < t_2 < ... < t_n \leq 1$, where the expression in brackets accounts for the approximation of the inner product between $f$ and $e_i$.

Fitting and representation of functional data is an important and intensively studied topic on its own, however, in this article we assume that datapoints are {\it fully observed}, i.e. we observe the whole curves, instead of just a sample of points. For more information on fitting we refer to \cite{ramsay:silverman:2005}.

\subsection{Statistics in Hilbert spaces}

Random function is the key concept used in the sequel. We can think of it as an extension of a random value or a random vector from $\bR^2$. Instead of drawing a random point from a plane, we draw a whole function from all functions in $L_2$. In this section we introduce a mean function and a covariance operator. We will extend the concept of expected value of a scalar variable to vectors and functions.

In the space $\bR^2$, we refer to a mean as an expected value on each coefficient, for example, having a random vector $\bX = (X_1,X_2)$, we define a mean of this variable as $\E \bX = (\E X_1, \E X_2)$.

Similarly we can look at a random function $X$. We can define it's mean as a mean on each coefficient. Let's take a fixed basis $\{ e_i \}_{1 \leq i}$ and an expansion of $X$, given by
\[
  X = \sum_{i=1}^\infty \ip{e_i}{X}e_i.
\]
Then, random variables $\ip{e_i}{X}$ correspond to coefficients, so a mean function can be defined as 
\[
  \E X = \sum_{i=1}^\infty (\E \ip{e_i}{X}) e_i,
\]
where we defined the expectation of random function as a sum of scalar expectations $\E \ip{e_i}{X}$. Note that, as in scalar case $\E \ip{e_i}{X}$ may not be finite and then also $\E X$ will not exist. Moreover, in order to have $\E X \in L_2$, we need the series $(\E \ip{e_i}{X})_{1 \leq i}$ to be square summable.

Defining covariance in functional space is more complicated and, in order to show a direct relation, we will recall a non-standard representation of a covariance in multidimensional space. One of the ways to look at a covariance of a random vector $\bX$ in $\bR^2$, is to see it as a linear function $\bR^2 \rightarrow \bR^2$, defined as
\[
C(v) = \E \ip{\bX}{v} \bX,
\]
where $v \in \bR^2$. We have $C(v) = (\E \bX \bX') v$ (and we refer to the variance as $C = \E \bX\bX'$). Note, that since $\ip{\bX}{v}$ is a scalar, $\ip{\bX}{v} \bX$ is a random vector and therefore we just use the definition of the mean vector. 

In the functional space, we will define a covariance operator as an operator $L_2 \rightarrow L_2$. For a random element $X$ in $L_2$, similarly to the vector case, we take the expectation of $\ip{\bX}{v} \bX$, i.e.
\[
C_X(f) = \E \ip{X}{f} X,\ \ \text{where } f \in L_2.
\]
We call $\ip{X}{\ \cdot\ } X$, an outer product and write $X \otimes X = \ip{X}{\ \cdot\ } X$. Again, we consider only such random functions $X$ that $C_X$ exists. One can show that the operator $C_X$ is a positively definite, symmetric Hilbert--Schmidt operator and therefore it can be inversed.

\section{Functional Time Series}

In many practical situations functions are naturally ordered in time. For example, when we deal with daily observations of the stock market or with sequences of tumor scans. Then, we are in the context of so--called functional time series (FTS).

\begin{figure}
\includegraphics[height=10cm,width=6cm,angle=90]{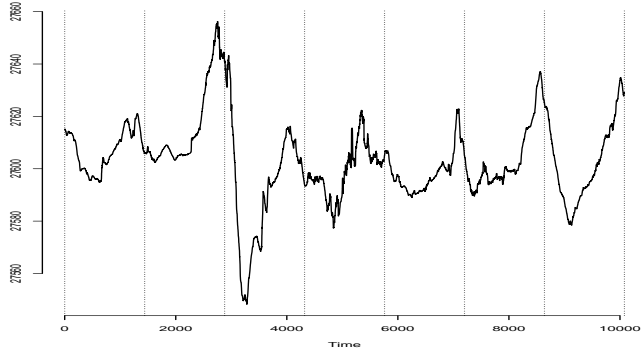} \centering
\caption{Horizontal component of the magnetic field measured in
one minute resolution at Honolulu magnetic observatory from 1/1/2001
00:00 UT to  1/7/2001 24:00 UT. 1440 measurements per day.}\label{fig:honolulu}
\end{figure}

As a motivating example consider Figure \ref{fig:honolulu}. Here, the assumption of independence can be too strong -- values at the beginning of each day are highly correlated with those at the end of the preceding day. Moreover, we see that big jumps are often followed by significant drops.

These, and similar features, may indicate significant temporal dependence not just within a subject, but also between different subjects (e.g.\ days). In this section we discuss possible frameworks which allow to quantify, and use this additional information.

\subsection{Stationarity}

Many physical processes are known to have time-invariant distribution. This motivates the frequentist approach to time series, where we assume that the structure does not change in time and we interfere from estimated covariances.

Let $\{X_t\}$ be a series of random functions. We say that $\{X_t\}$ is {\it stationary} in the {\it strong--sense} if for any $h \in \bZ, k \in \bN$ and any sequence of indices $t_1,t_2,...,t_k$ vectors $(X_{t_1},...,X_{t_k})$ and $(X_{h + t_1},...,X_{h + t_k})$ are identically distributed.

We also define {\it weak stationary} by looking only on the second order structure of the series. We say that $\{X_t\}$ is {\it weakly stationary} if $E\|X_t\|^2<\infty$ and
\begin{enumerate}
\item $\E X_t = \E X_0$ for each $t \in \bZ$ and
\item $\E X_t \otimes X_s = \E X_{t-s} \otimes X_0$ for each $t,s \in \bZ$.
\end{enumerate}

Additionally, we will assume that a series in the sequel are {\it weakly dependent}, which intuitively means that observations from the {\it far past} have little to no effect on the present. Many frameworks were suggested to quantify this behavior, for a survey on most popular ones refer to \cite{kidzinski}. For simplicity in this work we just assume a very strong condition, of a weakly dependent series as a stationary series for which
$$ \sum_{t=0}^\infty \E \|X_t \otimes X_0\|^2 < \infty, $$
meaning that the covariance between far elements decays very fast.

\subsection{Functional linear regression}
\label{sec:linreg}

One of the most popular frameworks in classical statistics is the linear regression, where we try to quantify the linear dependence between two (possibly multivariate) variables $X$ and $Y$. The problem of finding the relation of this type can be also addressed in FDA. As an illustrative example we can think of the relation between some farm's income during the year (a function of time defined on a yearly interval) and precipitation over a year.

We assume the model
\begin{equation}
Y_t = A(X_t) + \varepsilon_t,\quad t\geq 1, \label{eq:linmod}
\end{equation}

where $A$ is a linear Hilbert-Schmidt operator from $L_2 \rightarrow L_2$ and $\varepsilon_t$ is a sequence of independent identically distributed functions white noise sequence, independent from $(X_t)$.

As we are concerned with functional time series, we will assume that $X_t, Y_t$ are weakly stationary and weakly dependent. Classical case of iid $X_t$ is of great scientific interest and the interested reader is referred to \cite{ramsay:silverman:2005} and \cite{yao:muller:wang:2005AS}. 

Although the functional linear regression shares many properties with its multivariate equivalent, there are important theoretical difficulties, which preclude direct extension of the results from the simpler setup. Especially, we note that the linear operator $A : L_2 \rightarrow L_2$ is infinitely dimensional, which considerably complicates the estimation. If we approach the problem in the classical way by multiplying both sides of \eqref{eq:linmod} by $X_t$ and taking the expectation, for $t \geq 1$ we have
\begin{align*}
\E Y_t \otimes X_t = \E A(X_t)\otimes X_t + \E \varepsilon_t \otimes X_t = \E A(X_t)\otimes X_t = A(\E X_t\otimes X_t),
\end{align*}
by independence of $X_t$ and $\varepsilon_t$. Now, for convenience, let's denote this by
\begin{equation}
C^{XY} = A C^X, \label{eq:varcov}
\end{equation}
where $C^{XY}$ is the cross-covariance operator of $X$ and $Y$ and $C^X$ is the covariance of $X$. Now, the natural way to obtain $A$ is to apply the inverse of $C^X$ to both sides of the equation \eqref{eq:varcov}, which yields
\[
A = C^{XY} (C^X)^{-1}.
\]
The main problem is that the operator $(C^X)^{-1}$ is no longer bounded. Indeed, the domain of $C^X$ is only a subset $D$, say, of $L_2$. To see this, note that formally, as the inverse of $C^X$ is a linear operator, we may express
$(C^X)^{-1}(x)=\sum_{k\geq 0}\lambda_k^{-1}\langle e_k,x\rangle e_k$, where $\lambda_k$ and $e_k$ are the eigenvalues (tending to zero) and eigenfunctions of $C^X$. Hence, $D=\{x\in L_2 \colon \sum_{k\geq 1}\langle x,e_k\rangle^2\lambda_k^{-2}<\infty\}$.
The problem can be approached by some regularization. E.g.\ one may replace $(C^X)^{-1}$ by a finite dimensional approximation of the form $\sum_{k\leq K}\lambda_k^{-1}e_k\otimes e_k$, where $K$ is a tuning parameter. This is still quite delicate, when applied to the sample version. Then for large values of $K$, if we underestimate one of the small eigenvalues, its reciprocal explodes and will lead to very instable estimators. On the other hand, for small $K$ we may get a very poor approximation of $A$.

This difficulty was addressed by Bosq \cite{bosq:2000}, who gives an extensive survey on the problem. However, proposed results are based on strong assumptions on the rate of convergence of eigenvalues, which are impossible to check in practice. Alternative, elementary data--driven approach, was suggested in \cite{kidzinski}.

Finally, note that exactly the same technique can be used for lagged linear regression, i.e. where the response $Y_t$ depends linearly not only on the current observation $X_t$ but also on the whole series $X_t$. Consider
\begin{equation}\label{eq:lagged}
Y_t = \sum_{k=0}^m A_k(X_{t-k}) + \varepsilon_t,
\end{equation}
where $m \in \bN$. Again, as an example, we can think of it as an income during a given year based on precipitation in last 3 years. 

Let $m$ be the largest lag that we want to take into account and let us introduce $Z_t = (X_{t},X_{t-1},...,X_{t-m}) \in L_2^m$. One can easily show that the space $L_2^m$ is a Hilbert space. Then, the model can be written as
\begin{equation}
Y_t = B Z_t + \varepsilon_t,\label{eq:reduced-lagged}
\end{equation}
where $B : L_2^m \rightarrow L_2$ is a linear operator such that $B Z_t = \sum_{k=0}^m A_k(Z_{t}^{(k)})$.

Now, for estimating $B$ in \eqref{eq:reduced-lagged}, we can apply the same estimation procedures as in \eqref{eq:linmod}. This method of estimation in lagged regression models is efficient only for small dimensions and small $m$, as opposed to the technique briefly introduced in the following section.

\subsection{Frequency-domain methods}

The lagged linear model \eqref{eq:lagged} can be linked with the concept of linear filtering, popular in multivariate time series as well as in signal processing. For the theory and survey on applications in this context we refer to the classical book of Oppenheim and Schafer \cite{oppenheim1989discrete}.

\begin{defi}
We say that $\cA = \{ A_k \}_{k \in \bZ}$ is a linear filter if for each $k \in \bZ$, $A_k \in L_2 \rightarrow L_2$ is a linear operator and $\sum \|A_t\|^2 < \infty$.
\end{defi}

In order to find a method for estimation of operators $A_t$ in \eqref{eq:lagged}, in a more efficient way than in \eqref{sec:linreg}, we employ Fourier analysis and results from the seminal work of Brillinger \cite{brillinger:1975}.

The Fourier transform has two important properties which simplify analysis of the process \eqref{eq:linmod}. First, multiplication in the frequency domain is equivalent to convolution in the time domain. Second, the Fourier transform is a bijection, so results in frequency domain are equivalent to these in the time domain.

To illustrate the usage of these features let us multiply equation \eqref{eq:linmod} by $X_s$ for some $s \in \bZ$ and take the expectation. By linearity we have
\begin{align*}
\E Y_t \otimes X_s &= \sum_{k \in \bZ} A_k\E X_{t - k} \otimes X_s,
\end{align*}
and by stationarity
\begin{align*}
\E Y_{u} \otimes X_0 &= \sum_{k \in \bZ} A_k\E X_{u - k} \otimes X_0,
\end{align*}
where $u = t - s$. Now, noting that on left we have $C^{YX}_u$ and on right we have the convolution of $A_k$ and $ C^{YX}_{u}$, the Fourier transform of both sides yields the so--called cross-spectral operator between $\{Y_t\}$ and $\{X_t\}$ and can be obtained as
\begin{align}
\cF^{YX}_\theta &= \cA(\theta) \cF^{X}_\theta, \label{eq:fundamental}
\end{align}
where $\cA(\theta)=\sum_{k\in\mathbb{Z}}A_ke^{\ii k\theta}$ is the \emph{frequency response function} of the series $\{ A_k \}_{k \in \bZ}$, $\cF^{YX}_\theta = \frac{1}{2\pi}\sum_{k\in\mathbb{Z}}(C^X)^ke^{-\ii k\theta}$ is the \emph{spectral density operator} of $\{Y_t\}$ and $\{X_t\}$ and $\cF^{X}_\theta$ is the spectral density operator of $\{X_t\}$.

Having relation \eqref{eq:fundamental}, again we can invert $\cF^{X}_\theta$ and obtain a closed--form expression for $\cA(\theta)$. Now, the inverse Fourier transform gives us coefficients of the model \eqref{eq:lagged}. Note that, although we employ more sophisticated tools than in linear regression \eqref{eq:varcov}, symbolically approach presented here is analogical, but developed in the frequency domain.

\bibliographystyle{plain}

\end{document}